\documentclass{elsart}
\usepackage{amsmath,amssymb,amsfonts,amsthm}
\usepackage{latexsym,mathrsfs}

\def\ad{\mathrm{ad\,}}

\usepackage{amssymb}
\usepackage{eucal}

\def\ad{\mathrm{ad\,}}

\newtheorem{lemma}{Lemma}
\def\openone{\leavevmode\hbox{\small1\kern-3.3pt\normalsize1}}

\def\bbbz{\Bbb{Z}}

\def\ad{\mbox{ad\,}}

\def\bbbc{{\Bbb C}}
\def\bbbr{{\Bbb R}}

\def\bbbz{{\Bbb Z}}
\def\openone{\leavevmode\hbox{\small1\kern-3.3pt\normalsize1}}

\usepackage{epsf}

\def\newpic#1{%
   \def\emline##1##2##3##4##5##6{%
      \put(##1,##2){\special{em:point #1##3}}%
      \put(##4,##5){\special{em:point #1##6}}%
      \special{em:line #1##3,#1##6}}}

\newpic{}

\def\bbbr{{\Bbb R}}

\def\bbbc{{\Bbb C}}

\def\bbbz{{\Bbb Z}}

\def\ad{\mbox{ad}\,}

\def\openone{\leavevmode\hbox{\small1\kern-3.3pt\normalsize1}}

\allowdisplaybreaks
\begin{document}

\begin{frontmatter}

\title{Solutions of multi-component NLS models and
spinor Bose-Einstein condensates }

\author{V. S. Gerdjikov, N. A. Kostov, T. I. Valchev}
\address{Institute for Nuclear Research and Nuclear
Energy,  Bulgarian Academy of Sciences, 72 Tsarigradsko chaussee
1784 Sofia, Bulgaria}

\maketitle

\begin{abstract}
A three-  and five-component nonlinear Schrodinger-type models,
which describe spinor  Bose-Einstein condensates (BEC's) with
hyperfine structures $F=1$ and $F=2$ respectively, are studied.
These models for particular values of the coupling constants are
integrable by the inverse scattering method. They are related to
symmetric spaces of ${\bf BD.I}$-type $\simeq {\rm SO(2r+1)}/{\rm
SO(2)\times SO(2r-1)}$ for $r=2$ and $r=3$. Using conveniently
modified Zakha\-rov-Shabat dressing procedure we obtain different
types of soliton solutions.
\end{abstract}

\begin{keyword}
Bose-Einstein condensates, integrable systems, soliton models
\end{keyword}

\end{frontmatter}

\section{Introduction}

The dynamics of spinor BECs is described by a three-component
Gross-Pitaevskii (GP) system of equations. In the one-dimensional
approximation  the GP system goes into the following multicomponent
nonlinear Schr\"{o}dinger (MNLS) equation in 1D $x$-space \cite{IMW04}:
\begin{eqnarray}\label{eq:1}
&& i\partial_{t} \Phi_{1}+\partial^{2}_{x} \Phi_{1}+2(|\Phi_{1}|^2
+2|\Phi_{0}|^2) \Phi_{1} +2\Phi_{-1}^{*}\Phi_{0}^2=0, \nonumber \\
&& i\partial_{t} \Phi_{0}+\partial^{2}_{x}
\Phi_{0}+2(|\Phi_{-1}|^2
+|\Phi_{0}|^2+|\Phi_{1}|^2) \Phi_{0} +2\Phi_{0}^{*}\Phi_{1}\Phi_{-1}=0,\\
&& i\partial_{t}\Phi_{-1}+\partial^{2}_{x}
\Phi_{-1}+2(|\Phi_{-1}|^2+ 2|\Phi_{0}|^2) \Phi_{-1}
+2\Phi_{1}^{*}\Phi_{0}^2=0. \nonumber
\end{eqnarray}
We consider BECs of alkali atoms in the $F=1$ hyperfine state,
elongated in $x$ direction and confined in the transverse
directions $y,z$ by purely optical means. Thus the assembly of
atoms in the $F=1$ hyperfine state can be described by a
normalized spinor wave vector $ {\bf\Phi}(x,t)=(\Phi_1(x,t),
\Phi_0(x,t), \Phi_{-1}(x,t))^{T}$ whose components are labeled by
the values of $m_F =1,0,-1$. The above model is integrable by
means of the inverse scattering transform method \cite{IMW04}. It also
allows an exact description of the dynamics and interaction of
bright solitons with spin degrees of freedom. Matter-wave solitons
are expected to be useful in atom laser, atom interferometry and
coherent atom transport. It could contribute to the realization of
quantum information processing or computation, as a part of new
field of atom optics.

Lax pairs and geometrical interpretation of the MNLS models
related to symmetric spaces (including  the model (\ref{eq:1}))
are given in \cite{ForKu*83}. Darboux transformation for this
special integrable model is developed in \cite{LLMML05}. In
\cite{uiw07} the authors study soliton solutions for the
multicomponent Gross-Pitaevskii equation for $F=2$ spinor
condensate by two different methods assuming single-mode
amplitudes and by generalizing Hirota's direct method for
multicomponent systems. They point out the importance of
integrable cases, which take place for particular choices of the
coupling constants.

The aim of present paper is to show that both systems mentioned
above are integrable by the inverse scattering method and are
related to  symmetric spaces \cite{Helg} ${\bf BD.I}$-type:
$\simeq {\rm SO(2r+1)}/{\rm SO(2)\times SO(2r-1)}$ with $r=2$ and
$r=3$ respectively. In Section 2 we formulate the Lax
representations for the models. Section 3 is devoted to the $F=2$
BEC model. In Section 4 we construct the fundamental analytic
solutions of the corresponding Lax operator $L$ and reduce the
inverse scattering problem (ISP) for $L$ to a Riemann-Hilbert problem (RHP).
Using the special properties of the ${\bf BD.I}$ symmetric spaces we also
obtain the minimal sets of scattering data $\mathfrak{T}_i$ each
of which allow one to reconstruct both the scattering matrix
$T(\lambda)$ and the corresponding potential $Q(x,t)$. This allows
us to derive in Section 5 their soliton solutions using suitable
modification of the Zakharov-Shabat dressing method, proposed in
\cite{ZMNP,Za*Mi}.

\section{Multicomponent nonlinear Schr\"{o}dinger equations for
{\bf BD.I.} series of symmetric spaces}

MNLS equations  for the {\bf BD.I.} series of symmetric spaces
(algebras of the type $so(2r+1)$ and $J$ dual to $e_1$) have the
Lax representation $[L,M]=0$ as follows
\begin{eqnarray}\label{eq:3.1}
L\psi (x,t,\lambda ) &\equiv & i\partial_x\psi + (Q(x,t) -
\lambda J)\psi  (x,t,\lambda )=0.\\ \label{eq:3.2}
M\psi (x,t,\lambda ) &\equiv & i\partial_t\psi + (V_0(x,t) +
\lambda V_1(x,t) - \lambda ^2 J)\psi  (x,t,\lambda )=0, \\
V_1(x,t)&=& Q(x,t), \qquad V_0(x,t) = i \ad_J^{-1} \frac{d Q}{dx}
+ \frac{1}{2} \left[\ad_J^{-1} Q, Q(x,t) \right].
\end{eqnarray}
where
\begin{equation}\label{vec1}
Q=\left(\begin{array}{ccc}  0 & \vec{q}^{T} & 0 \\
  \vec{p} & 0 & s_{0}\vec{q} \\  0 & \vec{p}^{T}s_{0} & 0 \\
\end{array}\right),\qquad J=\mbox{diag}(1,0,\ldots 0, -1).
\end{equation}
The $2r-1$-vectors $\vec{q}$ and $\vec{p}$ have the form
\[ \vec{q} = (q_2,\dots , q_r, q_{r+1}, q_{r+2},\dots , q_{2r})^T, \qquad
\vec{p} = (p_2,\dots , p_r, p_{r+1}, p_{r+2},\dots , p_{2r})^T,
 \]
while the matrix $s_0$ represents the metric involved in the definition
of $so(2r-1)$, therefore it is related to the metric $S_0$ associated
with $so(2r+1)$ in the following manner
\begin{eqnarray}
S_{0}= \sum_{k=1}^{2r+1}(-1)^{k+1}E_{k,2r+2-k} =
\left(\begin{array}{ccc}
0 & 0  & 1 \\  0 & -s_{0} & 0 \\
1 & 0 & 0 \\ \end{array} \right),
\qquad (E_{kn})_{ij}=\delta_{ik}\delta_{nj}
\end{eqnarray}
Next we will use
\begin{equation}\label{eq:Epm}
\vec{E}_1^\pm = ( E_{\pm (e_1-e_2)}, \dots ,
E_{\pm (e_1-e_r)}, E_{\pm e_1}, E_{\pm(e_1+e_r)}, \dots , E_{\pm (e_1+e_2)} ),
\end{equation}
We will use also the "scalar product"
\[ (\vec{q}\cdot \vec{E}_1^+) = \sum_{k=2}^{r}( q_k(x,t)E_{e_1-e_k} +
q_{2r-k+2}(x,t)E_{e_1+e_k}) +q_{r+1}(x,t)E_{e_1}  .\]
Then the generic form of the potentials $Q(x,t)$ related to these type
of symmetric  spaces is
\begin{equation}\label{eq:Q}
    Q(x,t) = (\vec{q}(x,t) \cdot \vec{E}_1^+) + (\vec{p}(x,t) \cdot \vec{E}_1^-) ,
\end{equation}
where $E_\alpha $ are the Weyl
generators of the corresponding Lie algebra (see \cite{Helg} for details)
and $\Delta_1^+$ is the set of all positive roots of $so(2r+1)$
such that $(\alpha,e_1)=1$. In fact $\Delta_1^+ =\{ e_1, \quad e_1\pm e_k,\quad k=2,\dots,r\}$.

In terms of these notations the generic MNLS type equations connected
to ${\bf BD.I.}$ acquire the form
\begin{equation}\label{eq:4.2}
\begin{split}
i \vec{q}_t &+ \vec{q}_{xx} + 2 (\vec{q},\vec{p}) \vec{q} -
 (\vec{q},s_0\vec{q}) s_0\vec{p} =0, \\
i \vec{p}_t &- \vec{p}_{xx} - 2 (\vec{q},\vec{p}) \vec{p} -
(\vec{p},s_0\vec{p}) s_0\vec{q} =0,
\end{split}
\end{equation}
In the case of $r=2$ if we impose the reduction $p_k=q_k^{*}$ and introduce
the new variables $\Phi_1=q_{2}$, $\Phi_{0}=q_3/\sqrt{2}$, $\Phi_{-1}=q_4$ then
we reproduce the equations (\ref{eq:1}).

\section{F=2 spinor Bose-Einstein condensate, integrable case}

Let us introduce Hamiltonian for MNLS equations (\ref{eq:4.2})
with $\vec{p}=\epsilon \vec{q}^{*}$,
$\epsilon=\pm 1$
\begin{eqnarray}\label{eq:Ham1}
H_{{\rm MNLS}}=\int_{-\infty}^\infty d x
\left((\partial_{x}\vec{q},\partial_{x}\vec{q^{*}})-\epsilon
(\vec{q},\vec{q^{*}})^2+\epsilon
(\vec{q},s_0\vec{q})(\vec{q^{*}},s_{0}\vec{q^{*}})\right),
\end{eqnarray}
Define the number density and the singlet-pair amplitude by
\cite{CYHo,UK,uiw07}
\begin{eqnarray}
n=(\vec{{\bf \Phi}},\vec{{\bf \Phi^{*}}})=\sum_{\alpha=-2\ldots
2}\Phi_{\alpha}\Phi_{\alpha}^{*},\qquad \Theta=(\vec{{\bf
\Phi}},s_0\vec{{\bf \Phi}}).
\end{eqnarray}
where $\Phi_{2}=q_{2}$, $\Phi_{1}=q_{3}$, $\Phi_{0}=q_{4}$,
$\Phi_{-1}=q_{5}$,$\Phi_{-2}=q_{6}$. Then the singlet-pair
amplitude take the form \cite{CYHo,UK,uiw07}
\begin{eqnarray}
\Theta=2\Phi_{2}\Phi_{-2}-2 \Phi_{1}\Phi_{-1}+\Phi_{0}^2.
\end{eqnarray}
The physical meaning of $\Theta$ is a measure of formation of
spin-singlet "pairs" of bosons. The assembly of atoms in the $F=2$
hyperfine state can be described by a normalized spinor wave
vector
\begin{eqnarray}
{\bf\Phi}(x,t)=(\Phi_2(x,t), \Phi_1(x,t), \Phi_0(x,t),
\Phi_{-1}(x,t), \Phi_{-2}(x,t))^{T},
\end{eqnarray}
whose components are labeled by the values of $m_F =2,1,0,-1,-2$.
Here the energy functional within mean-field theory \cite{OM,Ho,CYHo,UK,uiw07}
is defined by
\begin{eqnarray}
\label{eq:GPe} E_{\mathrm{GP}}[{\bf \Phi}]= \int_{-\infty}^\infty
d x \left(\frac{\hbar^2}{2m}|\partial_x{\bf
\Phi}|^2+\frac{c_0}{2}n^2 +\frac{c_2}{2}{\bf
f}^2+\frac{c_4}{2}|\Theta|^2\right).
\end{eqnarray}
The coupling constants $c_i$ are real and can be expressed in
terms of a transverse confinement radius and a linear combination
of the $s$-wave scattering lengths of atoms
\cite{IMW04,imww04,uiw06} and ${\bf f}$ describe spin densities
\cite{uiw07}. Choosing $ c_2=0$, $c_4=1$ and $c_0=-2$ we obtain
integrable by the inverse scattering method model with the Hamiltonian.
We set for simplicity $\hbar=1, 2 m=1$ without any loss of generality. The
evolution equation is described by the multi-component
Gross-Pitaevskii equation in one dimension \cite{uiw07}
\begin{eqnarray}
\label{eq:GP} i\frac{\partial{\bf \Phi}}{\partial t}=\frac{\delta
E_{\mathrm{GP}}[{\bf \Phi}]}{\delta {\bf \Phi^*}}.
\end{eqnarray}
Then we have
\begin{eqnarray} \label{Eq}
i \vec{{\bf \Phi}}_t &+ \vec{{\bf \Phi}}_{xx} =-2\epsilon
(\vec{{\bf \Phi}},\vec{{\bf \Phi^{*}}}) \vec{{\bf \Phi}} +
\epsilon (\vec{{\bf \Phi}},s_0\vec{{\bf \Phi}}) s_0\vec{{\bf
\Phi^{*}}} ,
\end{eqnarray}
or in explicit form by components we have
\begin{eqnarray}
&&i\partial_t\Phi_{\pm 2}+\partial_{xx}\Phi_{\pm 2}= -2\epsilon
(\vec{{\bf \Phi}},\vec{{\bf \Phi^{*}}}) \Phi_{\pm 2} +\epsilon
(2\Phi_{2}\Phi_{-2}-2 \Phi_{1}\Phi_{-1}+\Phi_{0}^2) \Phi_{\mp
2}^*,\nonumber
\\
&&i\partial_t\Phi_{\pm 1}+\partial_{xx}\Phi_{\pm 1}= -2\epsilon
(\vec{{\bf \Phi}},\vec{{\bf \Phi^{*}}}) \Phi_{\pm 1} -\epsilon
(2\Phi_{2}\Phi_{-2}-2 \Phi_{1}\Phi_{-1}+\Phi_{0}^2) \Phi_{\mp
1}^*,\nonumber
\\
&&i\partial_t\Phi_{0}+\partial_{xx}\Phi_{0}=  -2\epsilon
(\vec{{\bf \Phi}},\vec{{\bf \Phi^{*}}}) \Phi_{\pm 0} +\epsilon
(2\Phi_{2}\Phi_{-2}-2 \Phi_{1}\Phi_{-1}+\Phi_{0}^2)
\Phi_{0}^*.\nonumber
\end{eqnarray}

\section{Inverse scattering method and reconstruction
of potential from minimal scattering data}\label{sec:3}

Herein we remind some basic features of the inverse scattering
theory appropriate for the special case of $F=2$ spinor BEC equations.

Solving the direct and the inverse scattering problem (ISP) for
$L$ uses the Jost solutions  which are defined by, see \cite{VSG2}
and the references therein
\begin{equation}
 \lim_{x \to -\infty} \phi(x,t,\lambda) e^{  i \lambda J x }=\openone, \qquad  \lim_{x \to \infty}\psi(x,t,\lambda) e^{  i \lambda J x } = \openone
 \end{equation}
 and the scattering matrix $T(\lambda,t)\equiv \psi^{-1}\phi(x,t,\lambda)$. Due to the special choice of $J$ and to the fact that the Jost solutions and the scattering
 matrix take values in the group $SO(2r+1)$ we can use the following block-matrix structure of $T(\lambda,t)$
\begin{equation}\label{eq:25.1}
T(\lambda,t) = \left( \begin{array}{ccc} m_1^+ & -\vec{b}^-{}^T & c_1^- \\
\vec{b}^+ & {\bf T}_{22} & -s_0\vec{b}^- \\ c_1^+ & \vec{b}^+{}^Ts_0 & m_1^- \\
\end{array}\right),
\end{equation}
where $\vec{b}^\pm (\lambda,t)$ are $2r-1$-component vectors,
${\bf T}_{22}(\lambda)$ is a $2r-1 \times 2r-1$ block and
$m_1^\pm (\lambda)$,  $c_1^\pm (\lambda)$ are scalar functions satisfying
$c_1^\pm = 1/2 (\vec{b}^\pm \cdot s_0 \vec{b}^\pm) /m_1^\pm$.

Important tools for reducing the ISP to a Riemann-Hilbert problem (RHP)
are the fundamental analytic solution (FAS) $\chi^{\pm}
(x,t,\lambda )$. Their construction is based on the generalized Gauss
decomposition of $T(\lambda,t)$
\begin{equation}\label{eq:FAS_J}
\chi ^\pm(x,t,\lambda)= \phi (x,t,\lambda) S_{J}^{\pm}(t,\lambda ) =
\psi (x,t,\lambda ) T_{J}^{\mp}(t,\lambda ) D_J^\pm (\lambda).
\end{equation}
Here $S_{J}^{\pm} $, $T_{J}^{\pm} $ upper- and lower-
block-triangular matrices, while $D_J^\pm(\lambda)$ are block-diagonal matrices with the same block structure
 as $T(\lambda,t)$ above. Skipping the details we give the explicit expressions of the Gauss factors in terms of the matrix elements of $T(\lambda,t)$
\begin{eqnarray}\label{eq:S_Jpm}
&&S_J^\pm (t,\lambda )= \exp \left( \pm (\vec{\tau}^\pm (\lambda,t) \cdot
\vec{E}_1^\pm ) \right),
\quad  T_J^\pm (t,\lambda )= \exp \left( \mp (\vec{\rho}^\pm (\lambda,t) \cdot
\vec{E}_1^\pm ) \right), \nonumber \\
&& D_J^+ = \left( \begin{array}{ccc} m_1^+ & 0 & 0 \\ 0 & {\bf m}_2^+ & 0 \\
0 & 0 & 1/m_1^+ \end{array} \right), \qquad  D_J^- =
\left( \begin{array}{ccc} 1/m_1^- & 0 & 0 \\ 0 & {\bf m}_2^- & 0 \\
0 & 0 & m_1^- \end{array} \right),
\end{eqnarray}
where $\vec{\tau}^\pm (\lambda ,t) = \vec{b}^\mp /m_1^\pm  $,
 $\vec{\rho}^\pm (\lambda ,t) = \vec{b}^\pm /m_1^\pm  $ and
\[  {\bf m}_2^+ = {\bf T}_{22} + \frac{\vec{b}^+ \vec{b}^-{}^T  }{m_1^+},
\qquad {\bf m}_2^- = {\bf T}_{22} +
 \frac{s_0\vec{b}^- \vec{b}^+{}^T s_0 }{m_1^-}.\]

If $Q(x,t) $ evolves according to (\ref{eq:1}) then the scattering
matrix and its elements satisfy the following linear evolution
equations
\begin{equation}\label{eq:evol}
i\frac{d\vec{b}^{\pm}}{d t} \pm \lambda ^2 \vec{b}^{\pm}(t,\lambda ) =0,
 \qquad  i\frac{d m_1^{\pm}}{d t}  =0,
 \qquad  i \frac{d{\bf m}_2^{\pm}}{d t}  =0,
\end{equation}
so the block-diagonal matrices $D^{\pm}(\lambda)$ can be considered as
generating functionals of the integrals of motion. The fact that
all $(2r-1)^2$ matrix elements of $m_2^\pm(\lambda)$ for
$\lambda \in \bbbc_\pm$  generate integrals of motion
reflect the superintegrability of the model and are due to the degeneracy
 of the dispersion law of
(\ref{eq:1}).
We remind that $D^\pm_J(\lambda)$ allow analytic extension for
$\lambda\in \bbbc_\pm$ and that their zeroes and
 poles determine the discrete eigenvalues of $L$.

The FAS for real $\lambda$ are linearly related
\begin{equation}\label{eq:rhp0}
\chi^+(x,t,\lambda) =\chi^-(x,t,\lambda) G_J(\lambda,t),
\qquad G_{0,J}(\lambda,t) =S^-_J(\lambda,t)S^+_J(\lambda,t) .
\end{equation}
One can rewrite eq. (\ref{eq:rhp0}) in an equivalent form for the FAS
$\xi^\pm(x,t,\lambda)=\chi^\pm (x,t,\lambda)e^{i\lambda Jx }$ which satisfy
also the relation
\begin{equation}\label{eq:rh-n}
\lim_{\lambda \to \infty} \xi^\pm(x,t,\lambda) = \openone.
\end{equation}
Then these FAS satisfy
\begin{equation}\label{eq:rhp1}
\xi^+(x,t,\lambda) =\xi^-(x,t,\lambda) G_J(x,\lambda,t), \qquad G_{J}(x,\lambda,t) =e^{-i\lambda Jx}G^-_{0,J}(\lambda,t)e^{i\lambda Jx} .
\end{equation}
Obviously the sewing function $G_j(x,\lambda,t)$ is uniquely determined by the
Gauss factors $S_J^\pm (\lambda,t)$.
In view of eq. (\ref{eq:S_Jpm}) we arrive to the following
\begin{lemma}\label{lem:ms}
Let the potential $Q(x,t)$ is such that the Lax operator $L$ has no discrete eigenvalues.
Then as minimal set of scattering data which determines uniquely the scattering matrix
$T(\lambda,t)$ and
the corresponding potential $Q(x,t)$ one can consider either one of the sets $\mathfrak{T}_i$, $i=1,2$
\begin{equation}\label{eq:T_i}
\mathfrak{T}_1 \equiv \{ \vec{\rho}^+(\lambda,t), \vec{\rho}^-(\lambda,t),
 \quad \lambda \in \bbbr\},
\qquad \mathfrak{T}_2 \equiv \{ \vec{\tau}^+(\lambda,t), \vec{\tau}^-(\lambda,t),
 \quad \lambda \in \bbbr\}.
\end{equation}
\end{lemma}

\begin{proof} i) From the fact that $T(\lambda,t)\in SO(2r+1)$ one can derive that
\begin{equation}\label{eq:25.3}
\frac{1}{m_1^+m_1^-} = 1 + (\vec{\rho^+},\vec{\rho^-}) + \frac{1}{4}
(\vec{\rho^+},s_0\vec{\rho^+}) (\vec{\rho^-}, s_0\vec{\rho^-})
\end{equation}
for $\lambda\in\bbbr$. Using the analyticity properties of $m_1^\pm $
 we can recover them from eq. (\ref{eq:25.3})
using Cauchy-Plemelji formulae. Given $\mathfrak{T}_i$ and $m_1^\pm$
one easily recovers $\vec{b}^\pm(\lambda)$
and $c_1^\pm(\lambda)$. In order to recover ${\bf m}_2^\pm$ one again
uses their analyticity properties,
 only now the problem reduces to a RHP for functions on $SO(2r+1)$.
 The details will be presented elsewhere.

ii) Obviously, given $\mathfrak{T}_i$ one uniquely recovers the sewing
function $G_J(x,t,\lambda)$.
In order to recover the corresponding potential $Q(x,t)$ one can use
the fact that the RHP (\ref{eq:rhp1})
with canonical normalization has unique solution. Given that solution
$\chi^\pm(x,t,\lambda)$ one recovers $Q(x,t)$ via the formula
\begin{equation}\label{eq:QQ}
Q(x,t) = \lim_{\lambda\to\infty} \lambda
\left( J - \chi^\pm J \widehat{\chi}^\pm(x,t,\lambda)\right).
\end{equation}
which is well known.
\end{proof}
We impose also the standard reduction, namely assume that  $Q(x,t)=Q^\dag(x,t)$,
 or in components $p_k=q_k^*$.
As a consequence we have $\vec{\rho}^-(\lambda,t)=\vec{\rho}^{+,*}(\lambda,t)$
 and $\vec{\tau}^-(\lambda,t)=\vec{\tau}^{+,*}(\lambda,t)$.

\section{Dressing method and soliton solutions}

The main goal of the dressing method \cite{1,I04,GGK05a,GGK05b,VSG1}
is, starting from a known solutions $\chi^\pm_0(x,t,\lambda)$ of
$L_0(\lambda) $ with potential $Q_{(0)}(x,t)$ to construct new singular
solutions $\chi^\pm_1(x,t,\lambda )$ of $L$ with a potential
$Q_{(1)}(x,t)$ with two additinal singularities located at prescribed
positions $\lambda _1^\pm $; the reduction $\vec{p} =\vec{q}^{*}$
ensures that $\lambda_1^-=(\lambda_1^+)^*$. It is related to the regular
one by a dressing factor $u(x,t,\lambda )$
\begin{equation}\label{eq:Dressfactor}
\chi^{\pm}_1(x,t,\lambda)=u(x,\lambda) \chi^{\pm}_0(x,t,\lambda)
u_{-}^{-1}(\lambda ). \qquad u_-(\lambda )=\lim_{x\to -\infty }
u(x,\lambda )
\end{equation}
Note that $u_-(\lambda )$ is a block-diagonal matrix. The dressing
factor $u(x,\lambda ) $ must satisfy the equation
\begin{equation}\label{eq:u-eq}
i\partial_x u + Q_{(1)}(x) u - u Q_{(0)}(x)- \lambda
[J,u(x,\lambda)] =0,
\end{equation}
and the normalization condition $\lim_{\lambda \to\infty }
u(x,\lambda ) =\openone $. Besides $\chi ^\pm_i (x,\lambda ) $,
$i=0,1 $ and $u(x,\lambda ) $ must belong to the corresponding Lie
group $SO(2r+1,\bbbc) $; in addition $u(x,\lambda ) $ by
construction has poles and zeroes at $\lambda_1^\pm$.

The construction of $u(x,\lambda ) $ is based on an appropriate
anzats specifying explicitly the form of its $\lambda
$-dependence \cite{Za*Mi,GGK05b} and the references therein.
\begin{equation}\label{eq:rank1}
u(x,\lambda)=\openone+(c(\lambda)-1)P (x,t) +\left(\frac{1}{c(\lambda)}-1\right)
 \overline{P}(x,t), \qquad
\overline{P} = S_0^{-1}P^TS_0,
\end{equation}
where  $P(x,t)$ and $\overline{P}(x,t)$ are projectors  whose rank $s$
can not exceed $r$ and which satisfy $P\overline{P}(x,t)=0$. Given a set of
$s$ linearly independent polarization vectors $|n_k\rangle $ spanning the
corresponding eigensubspase of $L$ one can define
\begin{equation}\label{eq:31}
\begin{split}
& P(x,t)=\sum_{a,b=1}^s |n_a(x,t)\rangle M_{ab}^{-1}\langle n_b^\dag(x,t)|,
\quad M_{ab}(x,t)= \langle n_b^\dag(x,t) | n_a(x,t)\rangle, \\
& |n_a(x,t)\rangle=\chi^+_0(x,t,\lambda^+)|n_{0,a}\rangle,\quad c(\lambda)=
\frac{\lambda-\lambda^+}{\lambda-\lambda^-}, \quad
\langle n_{0,a}|S_0|n_{0,b}\rangle =0.
\end{split}
\end{equation}
Taking the limit $\lambda \to \infty$ in eq. (\ref{eq:u-eq}) we get that
\[ Q_{(1)}(x,t) - Q_{(0)}(x,t) = (\lambda_1^- - \lambda_1^+) [ J,
P(x,t)- \overline{P}(x,t)] .\] Below we list the explicit expressions only
 for the one-soliton solutions. To this end we assume $Q_{(0)}=0$ and
 put $\lambda_1^\pm = \mu \pm i\nu$. As a result we get
\begin{equation}\label{eq:q-1s}
q_k^{(\rm 1s)}(x,t) = - 2i\nu \left( P_{1k}(x,t) + (-1)^k P_{\bar{k},2r+1} (x,t) \right),
\end{equation}
where $\bar{k}=2r+2-k$.

Repeating the above procedure $N$ times we can obtain $N$ soliton solutions.

\subsection{The case of rank one solitons} 

In this case $s=1$  so that  the generic (arbitrary $r$)
one-soliton solution reads
\begin{eqnarray}
q_{k}&=&\frac{-i\nu e^{-i\mu(x-vt-\delta_0)}}{\cosh 2z + \Delta_0^2 }
\left(\alpha_k e^{z-i\phi_k} +(-1)^k \alpha_{\bar{k}}
 e^{-z+i\phi_{\bar{k}}} \right), \nonumber\\
v & = & \frac{\nu^2-\mu^2}{\mu},\qquad u=-2\mu, \qquad z(x,t) =
 \nu (x - ut -\xi_0), \\
\xi_0 &=& \frac{1}{2\nu} \ln \frac{|n_{0,2r+1}|}{|n_{0,1}|},
\qquad \alpha_k =\frac{|n_{0,k}|}{\sqrt{|n_{0,1}||n_{0,2r+1}|}},\qquad
\Delta_0^2  =  \frac{\sum_{k=2}^{2r} |n_{0,k}|^2}{2|n_{0,1}n_{0,2r+1}|}, \nonumber
\end{eqnarray}
and $\delta_0=\arg n_{0,1}/\mu=-\arg n_{0,2r+1}/\mu$, $\phi_k = \arg n_{0,k}$.
The polarization vectors satisfy the following relation
\begin{equation}
\sum_{k=1}^r 2(-1)^{k+1} n_{0,k}n_{0,\bar{k}} + (-1)^r n_{0,r+1}^2 =0.
\end{equation}
Thus for $r=2$ we identify $\Phi_1 = q_2$, $\Phi_0 = q_3/\sqrt{2}$
and $\Phi_3 = q_4$ and we obtain the following solutions for the
equation (\ref{eq:1})
\begin{eqnarray}
\Phi_{\pm 1} & = & -\frac{2i\nu \sqrt{\alpha_2\alpha_4}
e^{-i\mu(x-vt-\delta_{\pm 1})}}{\cosh 2z + \Delta_0^2}
\left(\cos\phi_{\pm 1}\cosh z_{\pm 1}
-i\sin\phi_{\pm 1}\sinh z_{\pm 1}\right), \\
\delta_{\pm 1} &=& \delta_0 \mp\frac{\phi_2-\phi_4}{2\mu},
\qquad \phi_{\pm 1}=\frac{\phi_2+\phi_4}{2}
\qquad z_{\pm 1} = z\mp \frac{1}{2}\ln \frac{\alpha_4}{\alpha_2},\nonumber\\
\Phi_{0} & = &-\frac{\sqrt{2} i\nu \alpha_3e^{-i\mu(x-vt-\delta_0)}}
{\cosh 2z + \Delta_0^2} \left(\cos\phi_3\sinh z - i\sin\phi_3\cosh z
\right).
\end{eqnarray}

For $r=3$ we identify $\Phi_2 = q_2$, $\Phi_1 = q_3$, $\Phi_0 = q_4$,  $\Phi_{-1} = q_5$ and
$\Phi_{-2} = q_6$, so that  the one-soliton solution for equation (\ref{Eq}) reads
\begin{eqnarray}
\Phi_{\pm 2} & = & -\frac{2i\nu \sqrt{\alpha_2\alpha_6} e^{-i\mu(x-vt-\delta_{\pm 2})}} {\cosh 2z + \Delta_0^2} \left( \cos\phi_{\pm 2}
\cosh z_{\pm 2} - i\sin\phi_{\pm 2} \sinh z_{\pm 2} \right), \\
\Phi_{\pm 1} & = & -\frac{2i\nu \sqrt{\alpha_3\alpha_5} e^{-i\mu(x-vt-\delta_{\pm 1})}} {\cosh 2z + \Delta_0^2} \left( \cos\phi_{\pm 1}\sinh z_{\pm 1}
- i \sin\phi_{\pm 1}\cosh z_{\pm 1} \right),\\
\delta_{\pm 2} & = & \delta_0 \mp\frac{\phi_2-\phi_6}{2\mu},
\qquad\phi_{\pm 2} = \frac{\phi_2+\phi_6}{2}
\qquad z_{\pm 2} = z\mp \frac{1}{2}\ln\frac{\alpha_{6}}{\alpha_{2}},\nonumber\\
\delta_{\pm 1} & = & \delta_0 \mp\frac{\phi_3-\phi_5}{2\mu},
\qquad \phi_{\pm 1}=\frac{\phi_3+\phi_5}{2},\qquad
z_{\pm 1} = z\mp \frac{1}{2}\ln\frac{\alpha_5}{\alpha_3}, \nonumber  \\
\Phi_{0} & = & -\frac{2i\nu \alpha_4e^{-i\mu(x-vt-\delta_0)}}
{\cosh 2z + \Delta_0^2} \left(\cos \phi_4 \cosh z -
i \sin \phi_4 \sinh z \right).
\end{eqnarray}
Choosing appropriately the polarization vectors $|n\rangle $ we
are able to reproduce the soliton solutions obtained by Wadati et
al. both for $F=1$ and $F=2$ BEC.

\subsection{The case of rank two solitons }

Here $s=2$ and we have two linearly independent polarization
vectors $|n_a\rangle$, $a=1,2$. From eq. (\ref{eq:31}) we get
\begin{equation}\label{eq:r-2}
\begin{split}
P(x,t)&= \frac{1}{\det M} \left( |n_1(x,t)\rangle M_{22}\langle
n_1^\dag(x,t)| -|n_2(x,t)\rangle M_{12}\langle n_1^\dag(x,t)|
\right. \\
 &-\left . |n_1(x,t)\rangle M_{21}\langle n_2^\dag(x,t)| + |n_2(x,t)\rangle M_{11}\langle
 n_2^\dag(x,t)| \right),\\
\det M(x,t) &= M_{11} M_{22} - M_{12} M_{21}, \qquad
M_{ab}(x,t)= \langle n_a^\dag(x,t) | n_b(x,t)\rangle,
\end{split}
\end{equation}
The corresponding expressions for the rank 2 soliton solution are
obtained by inserting eq. (\ref{eq:r-2}) into (\ref{eq:q-1s}) and
are rather involved. We remark here that the reduction $Q^\dag =Q$
may not be sufficient to ensure that $\det M$ is positive for all
$x$ and $t$, so for certain choices of $|n_a\rangle $ we may have
singular solitons. These and other properties of the rank 2
soliton solutions will be analyzed elsewhere.

\section{Conclusions and discussion}
The main result of the present paper is that a special version of
the model describing $F=2$ spinor Bose-Einstein condensate is
integrable by the ISM. The corresponding Lax representation is
naturally related to the symmetric space ${\bf BD.I.}\simeq {\rm
SO(7)}/{\rm SO(2)\times SO(5)}$, see \cite{Helg}. For a generic
hyperfine spin $F$, the dynamics within the mean field theory is
described by the $2F+1$ component Gross-Pitaevskii equation in one
dimension. If all the spin dependent interactions vanish and only
intensity interaction exists, the multi-component Gross-Pitaevskii
equation in one dimension is equivalent to the vector nonlinear
Schr\"{o}dinger equation with $2F+1$ components \cite{ma74}.

Then equations (\ref{eq:4.2}) with the reduction $p=\epsilon
q^{*}, \epsilon=\pm 1$ are natural  generalization of the vector
nonlinear Schr\"{o}dinger equation, which adequately model the
spinor Bose-Einstein condensates for values of $F=r$ equal to 1
and 2. We expect that for generic $F$ these equations may be
useful in describing BECs with higher hyperfine structure.

Here we derived only generic one-soliton solutions. Following the
ideas of \cite{VG-DJ} one can classify different types of
one-soliton solutions related to different possible choices of the
rank of $P(x,t)$ and its polarization vectors. One can also derive
the $N$-soliton solutions by either repeating $N$ times the
dressing with $u$ (see eq. (\ref{eq:rank1}), or considering more
general dressing factors $u$ with $2N$ zeroes and poles in
$\lambda$. These and other problems will be addressed elsewhere.

\section*{Acknowledgments}\label{sec:Ack}

This work has been supported also by the National Science
Foundation of Bulgaria, contract No. F-1410.

\end{document}